\documentclass[12pt]{article}

\usepackage{amsmath}
\usepackage{amssymb}
\usepackage{amsthm}
\usepackage{hyperref}
\usepackage{enumerate}

\newtheorem{thm}{Theorem}[section]

\newtheorem{assumption}[thm]{Assumption}

\newtheorem{definition}[thm]{Definition}

\newtheorem{example}[thm]{Example}

\newtheorem{remark}[thm]{Remark}
\newenvironment{rem}{\begin{remark}\rm}{\end{remark}}

\newcommand{\E}{\mathbb{E}}

\title{An Elementary Proof of a Classical Information-Theoretic Formula}

\author{%
Xianming Liu$^{1}$, Ronit Bustin$^{2}$, Guangyue Han$^{3}$ and Shlomo Shamai$^{4}$\\
$^{1}$Huazhong University of Science and Technology, {\em email: xmliu@hust.edu.cn}\\
$^{2}$Technion-Israel Institute of Technology, {\em email: ronit.bustin@gmail.com}\\
$^{3}$The University of Hong Kong, {\em email: ghan@hku.hk}\\
$^{4}$Technion-Israel Institute of Technology, {\em email: sshlomo@ee.technion.ac.il}\\
}

\date{\today}

\begin{document} \maketitle

\begin{abstract}
A renowned information-theoretic formula by Shannon expresses the mutual information rate of a white Gaussian channel with a stationary Gaussian input as an integral of a simple function of the power spectral density of the channel input. We give in this paper a rigorous yet elementary proof of this classical formula. As opposed to all the conventional approaches, which either rely on heavy mathematical machineries or have to resort to some ``external'' results, our proof, which hinges on a recently proven sampling theorem, is elementary and self-contained, only using some well-known facts from basic calculus and matrix theory.
\end{abstract}

\section{Introduction} \label{Introduction}

Consider the following continuous-time white Gaussian channel
\begin{equation} \label{AWGN}
Y(t)=\int_0^t X(s) ds + B(t), \quad t \in \mathbb{R}^+,
\end{equation}
where $\{B(t): t \in \mathbb{R}^+\}$ denotes the standard Brownian motion, and the channel input $\{X(s): s \in \mathbb{R}\}$ is an independent stationary Gaussian process with power spectral density $f(\lambda)$. This paper is to give an elementary proof of the following classical information-theoretic formula (see, e.g., Theorem $6.7.1$
of~\cite{ih93})
\begin{equation} \label{Ihara}
\lim_{T \to \infty} \frac{1}{T} I(X_0^T; Y_0^T) = \frac{1}{4 \pi} \int_{-\infty}^{\infty} \log (1+ 2 \pi f(\lambda)) d \lambda.
\end{equation}

This renowned formula was first established by Shannon in his seminal work~\cite{shannon49} through a heuristic yet rather convincing spectrum-splitting argument and then treated more rigorously by numerous authors, predominantly using alternative channel formulations obtained via some orthogonal expansion representations in the relevant Hilbert space. Representative work in this direction include~\cite{huang-johnson-1962, huang-johnson-1963, hol64, ga68, ih93}, and the heart of all the approaches therein lies in a continuous-time version of the famed Szego's theorem (see, e.g., the theorem on page $139$ of~\cite{grenander58}). In a different direction, there have been efforts devoted to analyze continuous-time Gaussian channels using tools and techniques from stochastic calculus~\cite{du70, ka71, gu05, HanSong2016}, where the channel mutual information has been found to linked to an optimal linear filter. These links, together with well-known results from filtering theory~\cite{Wiener1949, YovitsJackson1955}, will conceivably recover (\ref{Ihara}).

It appear to us that all existing treatments either rely on heavy mathematical machineries or have to resort to some ``external'' results. By comparison, our proof, which hinges on a recently proven sampling theorem (Theorem $3.2$ in~\cite{LiuHan2018}), is elementary and self-contained, only using some well-known facts from basic calculus and matrix theory: it turns out that the aforementioned sampling theorem enables us to sidestep numerous complications that are otherwise present in the continuous-time regime and allows us to employ a spectral analysis of finite-dimensional matrices, rather than infinite-dimensional operators that some previous approaches would have to deal with. Moreover, as elaborated in Section~\ref{concluding-remarks}, our approach gives rise to a ``scalable'' version of Szego's theorem and naturally connects a continuous-time Gaussian channel to its sampled discrete-time versions, and thereby promising further applications in more general settings.

\section{A Heuristic Proof} \label{rough-ideas}

We first explain the aforementioned sampling theorem. For any given $T > 0$ and $n \in \mathbb{N}$, choose evenly spaced {\it sampling times} $t_i$, $i=0, 1, \dots, n$, such that $t_i~\footnote{Obviously each $t_i$ depends on $T$ and $n$, however we suppress the notational dependence for simplicity, which may apply to other notations in the paper as well.} = i T/n$ and let
$$
\Delta_{T, n} \triangleq \{t_0, t_1, \dots, t_n\}.
$$
Sampling the channel (\ref{AWGN}) over the time interval $[0, T]$ with respect to
$\Delta_{T, n}$, we obtain its sampled discrete-time version as
follows:
\begin{equation}  \label{after-sampling}
Y(t_{i})=\int_0^{t_{i}} X(s) ds + B(t_{i}), \quad i=0, 1, \ldots, n.
\end{equation}
Loosely speaking, Theorem $2.1$ in~\cite{LiuHan2018} says that as the above sampling gets increasingly finer, the mutual information of the discrete-time channel (\ref{after-sampling}) will converge to that of the original continuous-time channel (\ref{AWGN}).

Note that the mutual information of the channel (\ref{after-sampling}) can be computed as
\begin{align}  \label{mutual-information-det}
I(X_0^T; Y(\Delta_{T, n})) & =I(X_0^T; \{Y(t_{i})-Y(t_{i-1})\}_{i=1}^n) \nonumber\\
&=H(\{Y(t_{i})-Y(t_{i-1})\}_{i=1}^n) - H(\{B(t_{i})-B(t_{i-1})\}_{i=1}^n) \nonumber\\
&= \frac{1}{2} \log \det \left(I_n+ A_{T, n} \right),
\end{align}
where $I_n$ is the $n \times n$ identity matrix and $A_{T, n}$ is an $n \times n$ matrix whose $(i, j)$-th entry is defined as
$$
A_{T, n}(i, j)=\E\left[\frac{n}{T} \int_{t_i}^{t_{i+1}} X(s) ds \int_{t_j}^{t_{j+1}} X(s) ds\right].
$$
It then follows from the stationarity of $\{X(s)\}$ that
$$
A_{T, n}(i, j) = \gamma_{j-i},
$$
where, setting $t_{-i} = - i T/n$ for $i=1, 2, \dots, n-1$, we have
defined
$$
\gamma_l \triangleq \E\left[\frac{n}{T} \int_{t_0}^{t_{1}} X(s) ds
\int_{t_l}^{t_{l+1}} X(s) ds\right], \quad l=-(n-1), \dots, n-1.
$$
Noting that $A_{T, n}$ is a Hermitian (and Toeplitz) matrix and letting $\psi_1, \psi_2, \dots, \psi_n$ denote all its eigenvalues, we have
\begin{equation} \label{no-hat}
I(X_0^T; Y(\Delta_{T, n})) = \frac{1}{2} \sum_{m=1}^n \log (1+\psi_m).
\end{equation}

Now consider an $n \times n$ matrix $\hat{A}_{T, n}$ defined by
$$
\hat{A}_{T, n}(i, j) = \hat{\gamma}_{j-i},
$$
where $\hat{\gamma}_0 \triangleq \gamma_0$ and $\hat{\gamma}_l
\triangleq \gamma_l + \gamma_{n-l}$, $\hat{\gamma}_{-l} \triangleq
\gamma_{-l} + \gamma_{-n+l}$ for $l=1, 2, \dots, n-1$. Obviously,
$\hat{A}_{T, n}$ is an $n \times n$ circulant matrix whose
eigenvalues $\hat{\psi}_1, \hat{\psi}_1, \cdots, \hat{\psi}_1$ can
be readily computed as
\begin{equation} \label{hat-psi-m}
\hat{\psi}_m= \sum_{k=0}^{n-1} \hat{\gamma}_k e^{-2\pi i m k/n}.
\end{equation}
Now, for large $n$, approximating $\int_{t_k}^{t_{k+1}} X(s) ds$ by $X(t_k) \frac{T}{n}$, we have, for $0 < m < n/2$,
\begin{align}
\hat{\psi}_m & \approx \E[X^2(0)] \frac{T}{n}+\sum_{k=1}^{n-1} \E[X(t_0) X(t_k)] e^{-2\pi i (t_k-t_0) \frac{m}{T}} \frac{T}{n} + \sum_{k=1}^{n-1} \E[X(t_0) X(t_{n-k})] e^{-2\pi i (t_0-t_{n-k}) \frac{m}{T}} \frac{T}{n} \nonumber \\
& \approx 2\pi f(2\pi m/T), \label{approx-psd}
\end{align}
and for $n/2 < m < n$,
\begin{align}
\hat{\psi}_m & \approx \E[X^2(0)] \frac{T}{n}+\sum_{k=1}^{n-1} \E[X(t_0) X(t_k)] e^{-2\pi i (t_k-t_0) \frac{n-m}{T}} \frac{T}{n} + \sum_{k=1}^{n-1} \E[X(t_0) X(t_{n-k})] e^{-2\pi i (t_0-t_{n-k}) \frac{n-m}{T}} \frac{T}{n} \nonumber \\
& \approx 2\pi f(-2\pi (n-m)/T), \label{approx-psd}
\end{align}
where we have used the definition
$$
f(\lambda)= \frac{1}{2\pi} \int_{-\infty}^{\infty} R(\tau) e^{-i \tau \lambda} d\tau,
$$
where $R(\tau) = \E[X(0) X(\tau)]$ is the autocorrelation function of $\{X(s)\}$. Adapting some well-known arguments for establishing {\em aysmptotic equivalence} (see, e.g.,~\cite{grenander58} or~\cite{Gray2006}), we can prove that for large $T$ and large $n$,
\begin{equation} \label{no-hat-equals-hat}
\frac{\sum_{m=1}^n \log(1+\psi_m)}{T} \approx \frac{\sum_{m=1}^n \log(1+\hat{\psi}_m)}{T}.
\end{equation}
Now, collecting all the results above, we conclude that, for appropriately chosen large $T$ and large $n$,
\begin{align*}
\frac{1}{T} I(X_0^T; Y_0^T) & \stackrel{(a)}{\approx} \frac{1}{T} I(X_0^T; Y(\Delta_{T, n})) \\
& \stackrel{(b)}{=} \frac{1}{2} \log \det \left(I_n+ A_{T, n} \right)\\
& \stackrel{(c)}{=} \frac{\sum_{m=1}^n \log (1+\psi_m)}{2T}\\
& \stackrel{(d)}{\approx} \frac{\sum_{m=1}^n \log (1+\hat{\psi}_m)}{2T}\\
& \stackrel{(e)}{\approx} \sum_{m=-n/2}^{n/2} \log (1+2\pi f(2 \pi m/T)) \frac{1}{2T} \\
& \stackrel{(f)}{\approx} \frac{1}{4\pi}\int_{-\infty}^{\infty} \log (1+2 \pi f(\lambda)) d \lambda,
\end{align*}
where $(a)$ follows from Theorem~\ref{sampling-theorem-1}, $(b)$ follows from (\ref{mutual-information-det}), $(c)$ follows from (\ref{no-hat}), $(d)$ follows from (\ref{no-hat-equals-hat}), $(e)$ follows from (\ref{approx-psd}) and $(f)$ follows from the definition of the integral, establishing the formula (\ref{Ihara}).

The above proof is by no means rigourous, but, as elaborated in the next section, a refinement with some elementary $\varepsilon$-$\delta$ arguments and Fourier analysis arguments will certainly make it so to reach (\ref{Ihara}), which yields a rigorous proof of the classical formula.

\section{A Rigorous Proof}

First of all, we rigorously state our theorem.
\begin{thm} \label{main-theorem}
Assume that both $f(\lambda)$ and $R(\tau)$ are Lesbegue integrable over $\mathbb{R}$. Then,
\begin{equation}
\lim_{T \to \infty} \frac{1}{T} I(X_0^T; Y_0^T) = \frac{1}{4 \pi} \int_{-\infty}^{\infty} \log (1+ 2 \pi f(\lambda)) d \lambda.
\end{equation}
\end{thm}

\begin{rem}
It is well known that $f(\lambda)$ and $R(\tau)$ are a Fourier transform pair, and the integrability of one implies that the other one is uniformly bounded and uniformly continuous over $\mathbb{R}$. Moreover, it is easy to verify that $f(\lambda)$ is non-negative, and both $f(\lambda)$ and $R(\tau)$ are symmetric.
\end{rem}

We next state the sampling theorem that will be used in our proof, which is a weakened version of Theorem $2.1$ in~\cite{LiuHan2018} that holds true in a more general setting where sampling times may not be evenly spaced, and moreover, feedback and memory are possibly involved.
\begin{thm} \label{sampling-theorem-1}
For any fixed $T > 0$ and any sequence $\{\Delta_{T, n_k}: k \in \mathbb{N}\}$ satisfying $\Delta_{T, n_k} \subset \Delta_{T, n_{k+1}}$ for any feasible $k$, we have
$$
\lim_{k \to \infty} I(X_0^T; Y(\Delta_{T, n_k}))=I(X_0^T; Y_0^T),
$$
where $Y(\Delta_{T, n_k}) \triangleq \{Y(t_{0}), Y(t_{1}), \ldots, Y(t_{n_k})\}$.
\end{thm}

We are now ready to give the proof of our main result.

\begin{proof}[Proof of Theorem~\ref{main-theorem}]
Our proof consists of the following several steps.

{\bf Step 1.} In this step, we show that both $\|A_{T, n}\|_2$ and $\|\hat{A}_{T, n}\|_2$ are bounded from above uniformly over all $T > 0$ and $n \in \mathbb{N}$, namely, there exists $C > 0$ such that for all $T > 0$ and $n \in \mathbb{N}$, $\|A_{T, n}\|_2, \|\hat{A}_{T, n}\|_2 \leq C$. Here $\|\cdot\|_2$ denotes the operator norm induced by the $L^2$-norm for vectors.

It is straightforward to verify ({\it cf.} the proof of Lemma $4.1$
in~\cite{Gray2006}) that
$$
\|A_{T, n}\|_2 = \sup_{x \in \mathbb{R}^n: \|x\|_2=1} x A_{T, n} x^t
\leq \|g_{T, n}(\theta)\|_{\infty},
$$
where $\|\cdot\|_{\infty}$ denotes the $L^{\infty}$-norm and
$$
g_{T, n}(\theta) \triangleq \sum_{l=-(n-1)}^{n-1} \gamma_l e^{i l
\theta}.
$$
So, to establish the uniform boundedness of $\|A_{T, n}\|_2$, it
suffices to prove that $|g_{T, n}(\theta)|$ is bounded from above
uniformly all $\theta$, $T$ and $n$. Towards this end, we note that for any
feasible $l_1, l_2$,
\begin{align} \label{l1l2}
\sum_{l=l_1}^{l_2} |\gamma_l| &= \frac{n}{T} \left|\sum_{l=l_1}^{l_2} \E\left[ \int_{t_0}^{t_1} X(s) ds \int_{t_l}^{t_{l+1}} X(s) ds \right]\right| \nonumber \\
                  &= \frac{n}{T} \left|\sum_{l=l_1}^{l_2} \int_{t_0}^{t_1} \int_{t_l}^{t_{l+1}} \E[X(u) X(v)] dv du \right| \nonumber \\
                  &\leq \frac{n}{T} \int_{t_0}^{t_1} \left(\sum_{l=l_1}^{l_2} \int_{t_l}^{t_{l+1}} |R(v-u)| dv \right) du \nonumber \\
                  &\leq \frac{n}{T} \int_{t_0}^{t_1} \left(\int_{t_{l_1}}^{t_{l_2+1}} |R(v-u)| dv \right) du \nonumber \\
                  &\leq \frac{n}{T} \int_{t_0}^{t_1} \left(\int_{t_{l_1-1}}^{t_{l_2+1}} |R(\tau)| d\tau \right) du \nonumber \\
                  &= \int_{t_{l_1-1}}^{t_{l_2+1}} |R(\tau)| d\tau,
\end{align}
which immediately implies that
\begin{equation} \label{sum-of-gamma}
\sum_{l=-(n-1)}^{n-1} |\gamma_l| \leq \int_{-\infty}^{\infty}
|R(\tau)| d\tau,
\end{equation}
which, together with (\ref{hat-psi-m}), further implies that for any $m$,
$$
\hat{\psi}_m \leq 2 \int_{-\infty}^{\infty} |R(\tau)| d\tau.
$$
Note that a similar argument as above yields that for all $\theta$, $T$ and $n$,
$$
|g_{T, n}(\theta)| = \frac{n}{T} \left|\sum_{l=-(n-1)}^{n-1}
\E\left[ \int_{t_0}^{t_1} X(s) ds \int_{t_l}^{t_{l+1}} X(s) ds e^{i
l \theta}\right]\right| \leq \int_{-\infty}^{\infty} |R(\tau)|
d\tau,
$$
which implies the uniform boundedness of $\|A_{T, n}\|_2$, and
moreover, together with (\ref{hat-psi-m}), that of $\|\hat{A}_{T,
n}\|_2$.

{\bf Step 2.} In this step, we show that both $\|A_{T, n}\|_F^2/T$
and $\|\hat{A}_{T, n}\|_F^2/T$ are bounded from above uniformly over
all $T > 0$ and $n \in \mathbb{N}$. Here $\|\cdot\|_F$ denotes the Frobenious norm.

To prove the uniform boundedness of $\|A_{T, n}\|_F^2/T$, note that
$$
\frac{\|A_{T, n}\|_F^2}{T} = \frac{\sum_{k=-(n-1)}^{n-1} (n-|k|)
\gamma_k^2}{T} \leq \sum_{k=-(n-1)}^{n-1} (1-k/n) |\gamma_k|
\frac{\max\{|\gamma_k|\}}{T/n} \leq \int_{-\infty}^{\infty}
|R(\tau)| d\tau \,\, \|R\|_{\infty},
$$
where, for the last inequality, we have used (\ref{sum-of-gamma}).

A similar argument can be used to establish the uniform boundedness of $\|\hat{A}_{T, n}\|_F^2/T$.

{\bf Step 3.} In this step, we show that one can first fix a large
enough $T$ and then choose a large enough $n$ such that $\|A_{T,
n}-\hat{A}_{T, n}\|_F^2/T$ is arbitrarily small; more precisely, for any
$\varepsilon > 0$, there exists $T_0 > 0$ such that for any $T \geq
T_0$, there exists $n_0 > 0$ such that for all $n \geq n_0$,
$\|A_{T, n}-\hat{A}_{T, n}\|_F^2/T \leq \varepsilon$.

Towards this goal, we first note that
$$
\|A_{T, n}-\hat{A}_{T, n}\|_F^2 = \sum_{k=-(n-1)}^{n-1} (n-|k|) (\gamma_k-\hat{\gamma}_k)^2 \leq 2 \sum_{k=1}^{n-1} (n-k) \gamma_{n-k}^2 = 2 \sum_{k=1}^{n-1} k \gamma_{k}^2.
$$
In light of the integrability of $R(\tau)$, for any given $\varepsilon' > 0$, there exists $\tau_0 > 0$ such that
\begin{equation}
\int_{\tau_0}^{\infty} |R(\tau)| d\tau < \varepsilon'.
\end{equation}
Now, it can be easily verified that for any given $\varepsilon' > 0$, we can first fix a large enough $T$ and then choose large enough $n > 0$ such that $t_{\lfloor \varepsilon' n \rfloor} \geq \tau_0$, and furthermore,
$$
\sum_{k=\lfloor \varepsilon' n \rfloor}^{\infty} |\gamma_k| \stackrel{(a)}{\leq} \int_{\tau_0-T/n}^{\infty} |R(\tau)| d\tau < \varepsilon' \mbox{ and } \sum_{k=1}^{\lfloor \varepsilon' n \rfloor-1} (k/n) |\gamma_{k}|^2 \leq \varepsilon',
$$
where we have used (\ref{l1l2}) in deriving (a). It then follows that for $T$ and $n$
as above,
\begin{align*}
\frac{\|A_{T, n}-\hat{A}_{T, n}\|_F^2}{T} & \leq \left(2 \sum_{k=1}^{\lfloor \varepsilon' n \rfloor} (k/n) |\gamma_{k}| + 2 \sum_{k=\lfloor \varepsilon' n \rfloor}^{n-1} (k/n) |\gamma_{k}| \right) \frac{\max\{\gamma_k\}}{T/n}\\
& \leq \left(2 \sum_{k=1}^{\lfloor \varepsilon' n \rfloor} (k/n) |\gamma_{k}| + 2 \sum_{k=\lfloor \varepsilon' n \rfloor}^{\infty} |\gamma_{k}| \right) \frac{\max\{\gamma_k\}}{T/n}\\
& \leq 4 \varepsilon' \|R\|_{\infty},
\end{align*}
establishing Step $3$.

{\bf Step 4.} In this step, fixing a polynomial $p(x)$, we show that for any $\varepsilon > 0$, there exists $T_0 > 0$ such that for any $T \geq T_0$, there exists $n_0 > 0$ such that for all $n \geq n_0$,
$$
\left| \frac{\sum_{m=1}^n p(\psi_m)}{T}- \frac{\sum_{m=1}^n p(\hat{\psi}_m)}{T} \right| \leq \varepsilon.
$$

To achieve this goal, it suffices to prove that, given any fixed $k$, for any $\varepsilon > 0$, one can first fix a large enough $T$ and then choose a large enough $n$ such that
$$
\left| \frac{\sum_{m=1}^n \psi_m^k}{T}- \frac{\sum_{m=1}^n \hat{\psi}_m^k}{T} \right| \leq \varepsilon,
$$
which is equivalent to
$$
\left| \frac{tr(A_{T, n}^k-\hat{A}_{T, n}^k)}{T} \right| \leq \varepsilon.
$$
First of all, we note that
$$
A_{T, n}^k-\hat{A}_{T, n}^k=(A_{T, n}^k-A_{T, n}^{k-1}\hat{A}_{T, n})+(A_{T, n}^{k-1}\hat{A}_{T, n}-A_{T, n}^{k-2}\hat{A}_{T, n}^2)+\dots+(A_{T, n} \hat{A}_{T, n}^{k-1}- \hat{A}_{T, n}^k).
$$
And for the first term, using the well-known fact that for any two compatible matrices $E_1, E_2$,
$$
(tr(E_1 E_2))^2 \leq \|E_1\|_F^2 \|E_2\|_F^2, \quad \|E_1 E_2\|_F^2 \leq \|E_1\|_2^2 \|E_2\|_F^2,
$$
we deduce that
\begin{align*}
\left( \frac{tr(A_{T, n}^k-A_{T, n}^{k-1} \hat{A}_{T, n})}{T} \right)^2 & = \left( \frac{tr(A_{T, n}^{k-1}(A_{T, n}-\hat{A}_{T, n})}{T} \right)^2 \\
                                                   & \leq \frac{\|A_{T, n}^{k-1}\|_F^2 \|A_{T, n}-\hat{A}_{T, n}\|_F^2}{T^2} \\
                                                   & \leq \|A_{T, n}\|_2^{2(k-2)} \frac{\|A_{T, n}\|_F^2}{T} \frac{\|A_{T, n}-\hat{A}_{T, n}\|_F^2}{T}.
\end{align*}
It then follows from Steps $1$, $2$ and $3$ that for any $\varepsilon' > 0$, one can first fix a large enough $T$ and then choose a large enough $n$ such that
$$
\left|\frac{tr(A_{T, n}^k-A_{T, n}^{k-1} \hat{A}_{T, n})}{T}\right| < \varepsilon'.
$$
A completely parallel argument can be used to establish the same statement for other terms, which in turn implies our goal in this step.

{\bf Step 5.} In this step, we finish the proof of the theorem. First of all, let $\{\varepsilon_k\}$ be a monotone decreasing sequence of positive real numbers convergent to $0$. For any $\varepsilon_k > 0$, we first arbitrarily choose a monotone increasing sequence $\{T_k\}$ of positive real numbers divergent to infinity, and then, applying Theorem~\ref{sampling-theorem-1}, choose $n_k$ for each $T_k$ such that
\begin{equation} \label{two-close-sequences}
\left|\frac{1}{T_k} I(M; Y_0^{T_k}) - \frac{1}{T_k} I(M;
Y(\Delta_{T, n_k}))\right| \leq \varepsilon_k.
\end{equation}
Then, applying the Weierstrass approximation theorem to the
continuous function $\log(1+x)/x$, we choose two polynomials
$p_1^{(k)}(x), p_2^{(k)}(x)$ such that for all $x \in [0, 2
\int_{-\infty}^{\infty} |R(\tau)| d\tau]$,
\begin{equation} \label{p1-p2}
p_1^{(k)}(x) \leq \log (1+x) \leq p_2^{(k)}(x) \mbox{  and  }
p_2^{(k)}(x)-p_1^{(k)}(x) \leq \varepsilon_k x,
\end{equation}
which obviously leads to
\begin{equation} \label{psi-bounds}
\frac{\sum_{m=1}^{n_k} p_1^{(k)}(\psi_m)}{T_k} \leq
\frac{\sum_{m=1}^{n_k} \log (1+\psi_m)}{T_k} \leq
\frac{\sum_{m=1}^{n_k} p_1^{(k)}(\psi_m)}{T_k},
\end{equation}
\begin{equation} \label{psi-hat-bounds}
\frac{\sum_{m=1}^{n_k} p_1^{(k)}(\hat{\psi}_m)}{T_k} \leq
\frac{\sum_{m=1}^{n_k} \log (1+\hat{\psi}_m)}{T_k} \leq
\frac{\sum_{m=1}^{n_k} p_1^{(k)}(\hat{\psi}_m)}{T_k}.
\end{equation}
Re-choosing a larger $T_k$ first and then a larger $n_k$ if
necessary, we have, by Step $4$,
\begin{equation} \label{bounds-equal}
\lim_{k \to \infty} \frac{\sum_{m=1}^{n_k} p_1^{(k)}(\psi_m)}{T_k} =
\lim_{k \to \infty} \frac{\sum_{m=1}^{n_k}
p_1^{(k)}(\hat{\psi}_m)}{T_k}, \quad \lim_{k \to \infty}
\frac{\sum_{m=1}^{n_k} p_2^{(k)}(\psi_m)}{T_k} = \lim_{k \to \infty}
\frac{\sum_{m=1}^{n_k} p_2^{(k)}(\hat{\psi}_m)}{T_k}.
\end{equation}
Now, as elaborated in Appendix~\ref{Polynomial-PSD-Proof}, one can
show that (again re-choosing $n_k$ for each $T_k$ if necessary),
\begin{equation} \label{Polynomial-PSD}
\lim_{k \to \infty} \frac{\sum_{m=1}^{n_k} p_1^{(k)}(\hat{\psi}_m)}{T_k} = \frac{1}{2 \pi} \int p_1^{(k)}(2 \pi f(x)) dx, \quad \lim_{k \to \infty} \frac{\sum_{m=1}^{n_k} p_2^{(k)}(\hat{\psi}_m)}{T_k}=\frac{1}{2 \pi} \int p_2^{(k)}(2 \pi f(x)) dx.
\end{equation}
And moreover, from (\ref{p1-p2}), we deduce that
$$
\int (p_2^{(k)}(2 \pi f(x)) - p_2^{(k)}(2 \pi f(x))) dx \leq 2 \pi
\varepsilon_k \int f(x) dx.
$$
This, together with the integrability of $f(\cdot)$, implies that
\begin{equation} \label{left-equal-right}
\lim_{k \to \infty} \int p_1^{(k)}(2 \pi f(x)) dx = \lim_{k \to
\infty} \int p_2^{(k)}(2 \pi f(x)) dx,
\end{equation}
which, together with (\ref{psi-bounds}), (\ref{psi-hat-bounds}),
(\ref{bounds-equal}) and (\ref{Polynomial-PSD}), implies that
\begin{equation} \label{110}
\hspace{-1cm} \lim_{k \to \infty} \frac{\sum_{m=1}^{n_k} \log
(1+\psi_m)}{T_k} = \lim_{k \to \infty} \frac{\sum_{m=1}^{n_k} \log
(1+\hat{\psi}_m)}{T_k}=\frac{1}{2 \pi} \lim_{k \to \infty} \int
p_1^{(k)}(2 \pi f(x)) dx.
\end{equation}
Finally, similarly as in the derivation of (\ref{left-equal-right}), using (\ref{p1-p2}) and the integrability of $f(\cdot)$, we
conclude that
$$
\lim_{k \to \infty} \int p_1^{(k)}(2 \pi f(x)) dx = \int \log(1+2 \pi f(x)) dx,
$$
which, together with (\ref{110}), (\ref{no-hat}) and (\ref{two-close-sequences}), implies that
$$
\lim_{k \to \infty} \frac{1}{T_k} I(M; Y_0^{T_k}) = \frac{1}{4 \pi} \int_{-\infty}^{\infty} \log (1+ 2 \pi f(\lambda)) d \lambda.
$$
The theorem then immediately follows from a typical subsequence argument, as desired.

\end{proof}

\section{Concluding Remarks} \label{concluding-remarks}

Some remarks about the approach employed in this work are in order.

First, echoing~\cite{LiuHan2018}, we emphasize that time sampling, which is the key ingredient in our approach, ensures the inheritance of causality in converting a continuous-time Gaussian channel to its discrete-time versions, which stands in contrast to the orthogonal expansion representation in some previous approaches that destroys the temporal causality in the conversion process.

More technically, we note that our proof of Theorem~\ref{main-theorem} has actually established that one can appropriately ``scale'' $\{T_k\}$ and $\{n_k\}$ with $T_k/n_k$ shrinking to $0$ as $k$ tends to infinity ({\it i.e.,} the sampling gets finer) such that
$$
\lim_{k \to \infty} \frac{\log \det \left(I_n+ A_{T_k, n_k} \right)}{T_k} = \frac{1}{2\pi} \int \log(1+2 \pi f(x)) dx.
$$
Such a result can be regarded as a ``scalable'' version of Szego's theorem, which seems to serve as a bridge connecting discrete-time and continuous-time Szego's theorems.

As argued above, we believe that, other than recovering a classical information-theoretic formula with an elementary proof, our approach promises further applications in more general settings, which, for instance, include possible extensions of the formula (\ref{Ihara}) to continuous-time Gaussian channels with feedback and memory~\cite{ih93}, or with multi-users~\cite{LiuHan2018}, or multi-inputs and multiple-outputs~\cite{Brandenburg1974}.

\bigskip

{\bf Acknowledgement.} We would like to thank Professor Shunsuke Ihara for insightful discussions and for pointing out relevant references.

\section{Appendices} \appendix

\section{Proof of (\ref{Polynomial-PSD})} \label{Polynomial-PSD-Proof}

To prove (\ref{Polynomial-PSD}), it suffices to prove that for any
$q=1, 2, \dots$,
\begin{equation} \label{Power-PSD}
\lim_{k \to \infty} \frac{\sum_{m=1}^{n_k} \hat{\psi}_m^q}{T_k} =
\frac{1}{2 \pi} \int (2 \pi f(x))^q dx.
\end{equation}

For illustrative purposes, we now prove (\ref{Power-PSD}) for the case that $q=2$ in great detail. First of all, by (\ref{hat-psi-m}), we have,
for any $m=1, 2, \dots, n_k$,
\begin{equation} \label{hat-psi-m-square}
\hat{\psi}_m^2= \sum_{j_1=0}^{n_k-1} \sum_{j_2=0}^{n_k-1} \hat{\gamma}_{j_1}
\hat{\gamma}_{j_2} e^{-2\pi i m (j_1+j_2)/n_k}
\end{equation}
and furthermore
\begin{align*}
\sum_{m=1}^{n_k} \hat{\psi}_m^2 & = \sum_{m=1}^{n_k} \sum_{j_1=0}^{n_k-1} \sum_{j_2=0}^{n_k-1} \hat{\gamma}_{j_1} \hat{\gamma}_{j_2} e^{-2\pi i m (j_1+j_2)/n_k}\\
                            & = \sum_{l=0}^{n_k-1} \sum_{j=0}^{n_k-1} \hat{\gamma}_{j_1} \hat{\gamma}_{j_2} \sum_{m=1}^{n_k}  e^{-2\pi i m (j_1+j_2)/n_k}\\
                            & \stackrel{(a)}{=} n_k (\hat{\gamma}_0^2 + \sum_{l=1}^{n_k-1} \hat{\gamma}_l \hat{\gamma}_{n_k-l})\\
                            & = n_k (\sum_{l=-(n_k-1)}^{n_k-1} \gamma_l^2+ 2 \sum_{l=1}^{n_k-1}
\gamma_l \gamma_{n_k-l}),
\end{align*}
where for (a), we have used the easily verifiable fact that if $l+j$ is equal to $0$ or $n_k$, then $\sum_{m=1}^{n_k}  e^{-2\pi i m (l+j)/n_k}$ is equal to $n_k$, and $0$ otherwise. Noting that
$$
\gamma_l = \E\left[ \frac{n_k}{T_k} \int_{t_0}^{t_1}
\int_{t_l}^{t_{l+1}} X(s) X(s') ds ds' \right] = \frac{n_k}{T_k}
\int_{t_0}^{t_1} \int_{t_l}^{t_{l+1}} R(s'-s) ds ds',
$$
with a routine continuity argument using the definition of integral, we arrive at
\begin{equation}  \label{part-1}
\lim_{k \to \infty} \frac{n_k \sum_{l=-(n_k-1)}^{n_k-1}
\gamma_l^2}{T_k} = \lim_{k \to \infty} \int_{-T_k}^{T_k} R^2(s) ds =
\int_{-\infty}^{\infty} R(s) R(-s) ds = \frac{1}{2 \pi} \int (2 \pi
f(\lambda))^2 d \lambda,
\end{equation}
where we have used the uniform boundedness and uniform continuity of
$R(\cdot)$ for the first equality, and the last equality follows
from the fact that $f^2(\cdot)$ and $R*R(\cdot)$ are a Fourier
transform pair. Moreover, using the absolute summability of
$\{\gamma_l\}$ and the fact that $\lim_{\tau \to \infty} R(\tau)=0$
(this follows from the Riemann-Lesbegue lemma), we have
\begin{equation} \label{part-2}
\lim_{k \to \infty} \frac{n_k(2 \gamma_1 \gamma_{n-1} + \dots+ 2
\gamma_{n-1} \gamma_1)}{T_k} = 0.
\end{equation}
It then follows from (\ref{part-1}) and (\ref{part-2}) that
$$
\lim_{k \to \infty} \frac{\sum_{m=1}^{n_k} \hat{\psi}_m^2}{T_k} =
\frac{1}{2 \pi} \int (2 \pi f(\lambda))^2 d \lambda,
$$
establishing (\ref{Power-PSD}) for the case that $q=2$.

We next prove (\ref{Power-PSD}) for a general $q \geq 2$. Since the arguments are more tedious than yet
completely parallel to the case that $q=2$, we only outline the major steps below. In a parallel manner as above, we have, for any  $q \geq 2$,
\begin{align} \label{S1S2}
\sum_{m=1}^{n_k} \hat{\psi}_m^q & = \sum_{m=1}^{n_k} \sum_{j_1=0}^{n_k-1} \sum_{j_2=0}^{n_k-1} \dots \sum_{j_q=0}^{n_k-1} \hat{\gamma}_{j_1} \hat{\gamma}_{j_2} \dots \hat{\gamma}_{j_q}  e^{-2\pi i m (j_1+j_2+\dots+j_q)/n_k} \nonumber\\
& =  \sum_{j_1=0}^{n_k-1} \sum_{j_2=0}^{n_k-1} \dots \sum_{j_q=0}^{n_k-1} \hat{\gamma}_{j_1} \hat{\gamma}_{j_2} \dots \hat{\gamma}_{j_q} \sum_{m=1}^{n_k} e^{-2\pi i m (j_1+j_2+\dots+j_q)/n_k} \nonumber \\
&= n_k \sum_{j_1+j_2+\dots+j_q = 0, \; n_k, \dots, \; (q-1) n_k} \hat{\gamma}_{j_1} \hat{\gamma}_{j_2} \dots \hat{\gamma}_{j_q} \nonumber \\
& = n_k (S_1+S_2),
\end{align}
where $S_1$ is the summation of all terms taking the form of $\gamma_{j_1} \gamma_{j_2} \dots \gamma_{j_q}$ satisfying $j_1+j_2+\dots+j_{q-1}=j_q$ and $S_2$ is the summation of all the ``remaining'' terms. Then, similarly as in deriving (\ref{part-1}), we deduce that
\begin{equation}  \label{general-part-1}
\lim_{k \to \infty} \frac{n_k S_1}{T_k} = R*R*\dots*R(0) = \frac{1}{2 \pi} \int (2 \pi
f(\lambda))^q d \lambda.
\end{equation}
And similarly as in deriving (\ref{part-2}), we deduce that
\begin{equation} \label{general-part-2}
\lim_{k \to \infty} \frac{n_k S_2}{T_k} = 0.
\end{equation}
Finally, (\ref{Power-PSD}) then follows from (\ref{S1S2}), (\ref{general-part-1}) and (\ref{general-part-2}), which in turn implies (\ref{Polynomial-PSD}), as desired.

\end{document}